\documentclass{article}

\title{On the Limits of Information Spread by Memory-less Agents} 

\author{Niccolò D'Archivio\footnote{COATI, INRIA d’Université Côte d’Azur, Sophia-Antipolis, France.} \quad and \quad Robin Vacus\footnote{Bocconi University, Bocconi Institute for Data Science and Analytics, Milan, Italy}}

\date{}

\usepackage{amsmath,amsthm,amsfonts}		
\usepackage{dsfont}
\usepackage{inputenc}
\usepackage[margin=1in]{geometry}
\usepackage{bbold}
\usepackage[ruled,vlined]{algorithm2e}
\usepackage{graphicx}
\usepackage[usenames,dvipsnames]{xcolor}
\RestyleAlgo{boxruled}

\usepackage{hyperref}						
\hypersetup{
	hidelinks,
    colorlinks = true,
    linkcolor = MidnightBlue,
    citecolor = MidnightBlue
}

\usepackage[nameinlink]{cleveref}
\crefname{claim}{Claim}{Claims}
\crefname{observation}{Observation}{Observations}
\crefname{equation}{Eq.}{Eqs.}
\crefname{algorithm}{Algorithm}{Algorithms}
\crefname{lemma}{Lemma}{Lemmas}
\crefname{proposition}{Proposition}{Propositions}
\crefname{corollary}{Corollary}{Corollaries}

\newtheorem{theorem}{Theorem}
\newtheorem{lemma}[theorem]{Lemma}
\newtheorem{claim}[theorem]{Claim}
\newtheorem{proposition}[theorem]{Proposition}
\newtheorem{corollary}[theorem]{Corollary}

\theoremstyle{definition}

\DeclareMathOperator{\bbN}{\mathbb{N}}
\DeclareMathOperator{\bbR}{\mathbb{R}}

\newcommand{\pa}[1]{\left( #1 \right)}

\newcommand*{\medcap}{\mathbin{\scalebox{1.5}{\ensuremath{\cap}}}}

\renewcommand{\Pr}{\mathbb{P}}
\newcommand{\E}{{\rm I\kern-.3em E}}

\newcommand{\polylog}{\mathrm{polylog}}

\newcommand{\binomial}{\mathrm{Binomial}}

\newcommand{\rinf}{r_\infty}
\newcommand{\prot}{\mathcal{P}}

\newcommand{\support}{\mathcal{S}}
\newcommand{\event}{\mathcal{E}}

\begin{document}

\maketitle

\begin{abstract}
We address the self-stabilizing bit-dissemination problem, designed to capture the challenges of spreading information and reaching consensus among entities with minimal cognitive and communication capacities.
Specifically, a group of~$n$ agents is required to adopt the correct opinion, initially held by a single informed individual, choosing from two possible opinions.
In order to make decisions, agents are restricted to observing the opinions of a few randomly sampled agents, and lack the ability to communicate further and to identify the informed individual.
Additionally, agents cannot retain any information from one round to the next.
According to a recent publication by Becchetti et al. in SODA (2024), a logarithmic convergence time without memory is achievable in the parallel setting (where agents are updated simultaneously), as long as the number of samples is at least~$\Omega(\sqrt{n \log n})$.
However, determining the minimal sample size for an efficient protocol to exist remains a challenging open question.
As a preliminary step towards an answer, we establish the first lower bound for this problem in the parallel setting. Specifically, we demonstrate that it is impossible for any memory-less protocol with constant sample size, to converge with high probability in less than an almost-linear number of rounds.
This lower bound holds even when agents are aware of both the exact value of~$n$ and their own opinion, and encompasses various simple existing dynamics designed to achieve consensus.
Beyond the bit-dissemination problem, our result sheds light on the convergence time of the ``minority'' dynamics, the counterpart of the well-known majority rule,
whose chaotic behavior is yet to be fully understood despite the apparent simplicity of the algorithm.
\end{abstract}

\section{Introduction}

Exploring the computational power and limits of well-chosen models -- 
ones that are simple enough for analytical tractability, and yet
relevant to specific biological scenarios --
can lead to insightful conclusion about the functioning of biological distributed systems~\cite{afek_biological_2011,feinerman_collaborative_2012,emek_stone_2013,ghaffari_distributed_2015,clementi_search_2021}.
In line with this approach, we consider the \textit{bit-dissemination} problem, introduced in \cite{boczkowski_minimizing_2017} in order to evaluate the possibility to solve two fundamental problems concurrently: reaching agreement efficiently, while ensuring that an information possessed initially by a single individual is propagated to the whole group.
In order to fit biological scenarios, the problem features extremely constrained communications.
Agents engage in random interactions with just a few individuals at a time, as in the $\mathcal{PULL}$ model.
Furthermore, they can only disclose their current decision and no other information, following an assumption introduced in~\cite{korman_early_2022} to model situations in which individuals do not actively communicate~\cite{danchin_public_2004}.

In this paper, we further restrict attention to \textit{memory-less} entities lacking the ability to perform computations over extended periods of time; or at least not in a sufficiently reliable manner.
In particular, this assumption precludes the possibility to maintain clocks and counters, or to estimate the tendency of the dynamics.
Biological ensembles for which it is plausible include ant colonies~\cite{feinerman_individual_2017}, slime molds, cells and bacteria~\cite{miller_quorum_2001} or even plants~\cite{trewavas_plant_2002}.

Although it is always hard to rule out the possibility that species make use of memory (especially in the case of social insects), this modeling choice remains applicable even without presupposing specific cognitive abilities, as many species often stick to simple behavioral rules when attempting to reach a consensus, such as quorum sensing~\cite{sumpter_consensus_2008} or alignment rules~\cite{vicsek_collective_2012}.
Furthermore, our goal is not to focus on one truly realistic model, but rather to find out what are the minimal requirements to perform certain tasks.

\paragraph*{Informal description of the problem.}
More precisely, we consider a group of~$n$ agents holding binary opinions.
One of the agents, referred to as the \textit{source}, knows what opinion is ``correct'' and remains with it at all times. 
Execution proceeds in discrete rounds.
We assume that agents have no memory of what happened in previous rounds, besides their current opinion.
In the \textit{parallel} setting, all non-source agents are activated simultaneously in every round, while in the \textit{sequential} setting, only one non-source agent, selected uniformly at random, is activated.
Upon activation, a non-source agent~$i$ samples a set~$S$ consisting of~$\ell$ other agents drawn uniformly at random (with replacement\footnote{Note that when~$\ell \ll n$, sampling with and without replacement are essentially equivalent, and we choose the former for the sake of convenience.}). Then, based only on the opinions of the agents in~$S$ and on its own opinion, Agent~$i$ may choose to adopt a new opinion.
In particular, Agent~$i$ does not know whether or not $S$ contains the source.
A protocol is successful if every non-source agent eventually adopts the correct opinion and remains with it forever.
Finally, a protocol must converge independently of the initial opinions of the agents (including the correct opinion), which can be thought of as being chosen by an adversary.

\paragraph*{Previous works.}
The main parameters of the problems are the activation pattern (which may be parallel or sequential) and the sample size~$\ell$.
To compare protocols across various settings, we are typically interested in their \textit{convergence time}, i.e., the total number of activations required to reach consensus w.r.t. the correct opinion. Here, we will express the convergence time in terms of \textit{parallel rounds}: one parallel round is made up of $n$ activations, which corresponds to 1 round in the parallel setting and $n$ rounds in the sequential setting (although the two settings are not equivalent).

The bit-dissemination problem without memory was first studied in \cite{becchetti_role_2023}, in the sequential setting, where nearly matching lower and upper bounds are given.
On the one hand, the authors show that no protocol can converge in less than~$\Omega(n)$ parallel rounds in expectation, regardless of the sample size. On the other hand, they show that the well-known Voter dynamics (\Cref{prot:voter} in \Cref{sec:dynamics}) achieves consensus in~$O(n \log^2 n)$ parallel rounds with high probability.
Since the Voter dynamics only needs the sample size to be~$1$, these results imply that $\ell$ is not a critical parameter in the sequential setting.

Later, the authors of~\cite{becchetti_minority_2024} show that the aforementioned lower bound does not hold in the parallel setting.
In fact, they show that the \textit{minority} dynamics (\Cref{prot:minority} in \Cref{sec:dynamics}) converges in $O(\log^2 n)$ parallel rounds w.h.p., as long as the sample size is at least~$\Omega(\sqrt{n \log n})$.
This finding reveals that the best convergence times achievable in the sequential versus parallel settings differ by an exponential factor.
The bit-dissemination problem is all the more interesting to study as it is one of the most natural ones exhibiting this property.

Part of the explanation behind this curious phenomenon is that the stochastic processes involved are of different mathematical natures, in one setting compared to the other.
In the sequential setting, the number of agents with opinion~$1$ may only vary by at most one unit in every round, since only one agent is activated at a time. Therefore, independently of the protocol being operated, the evolution of the system can always be described by a ``birth-death'' chain, i.e., a Markov chain whose underlying graph is a path of size~$n$. In fact, all proofs in~\cite{becchetti_role_2023} heavily rely on this observation.
In contrast, in the parallel setting, the process may jump from any configuration to any other -- albeit with extremely small probability.
On the one hand, this characteristic allows for fast convergence.
For instance, qualitatively speaking, the minority dynamics succeeds by first reaching a configuration in which an appropriate proportion of the agents hold the wrong opinion; after what all non-source agents, perceiving the same minority, simultaneously adopt the correct opinion.
On the other hand, it complicates the analysis of any protocol in the parallel setting, and even more so the task of deriving lower bounds.
Following on from these works, we are ultimately interested in the following question:

\begin{center}
    \textit{Is there any protocol achieving a poly-logarithmic convergence time in the parallel setting, \\ when the sample size is $o(\sqrt{n})$?}
\end{center}

Based on the findings in~\cite{becchetti_minority_2024}, the minority dynamics is a natural candidate for this task.
Despite its extreme simplicity, the conditions under which it is able to converge quickly have not been identified.
While the analysis in~\cite{becchetti_minority_2024} relies on a sample size of at least~$\Omega(\sqrt{n \log n})$, the authors do not provide a lower bound for this parameter, nor do they justify informally why this quantity is necessary.
Therefore, we also consider the following independent question.

\begin{center}
    \textit{What is the minimal sample size for which the minority dynamics \\ converges in poly-logarithmic time?}
\end{center}

As a first step towards answering these questions, we focus on the case that the sample size~$\ell$ is constant, or in other words, independent of~$n$.
Beyond analytical tractability, several reasons motivate this assumption.
First, many existing opinion dynamics that have been traditionally studied in the context of consensus are defined with a small (fixed) number of samples~\cite{becchetti_consensus_2020}. This is the case for the Voter dynamics and majority dynamics, but also for the undecided states dynamics, as well as all \textit{population protocols}~\cite{aspnes_introduction_2009}, in which agents interact by pairs.
In addition, protocols relying on a sample size that increases with~$n$ implicitly require the agents to have some knowledge about the size of the population, which is often undesirable or unrealistic in distributed systems.
Finally, real biological entities are most likely interacting with few of their conspecifics even when they are part of a larger group. As an illustration, it has been shown empirically that the movement of any bird in a flock depends mostly on its 6 or 7 closest neighbors, regardless of how many individuals are present in the vicinity~\cite{ballerini_interaction_2008,bialek_statistical_2012}, while other works point at a similar phenomenon in fishes~\cite{katz_inferring_2011}.

\subsection{Problem Definition} \label{sec:problem_definition}

We consider a finite set $I = \{1,\ldots,n\}$ of agents.
Let~$X_t^{(i)} \in \{0,1\}$ be the \textit{opinion} of Agent~$i$ in round~$t$.
We assume that Agent~$1$ is the \textit{source} and holds the \textit{correct} opinion throughout the execution.
Denoting the sample size by~$\ell$ (independent of~$n$), a protocol~$\prot$ is defined as a family of functions
\begin{equation*}
    g_n^{[b]} : \{0,\ldots,\ell\} \rightarrow [0,1].
\end{equation*}
For a given population size~$n$, an opinion~$b$, and a sample of opinions containing $k$ '1' out of~$\ell$, $g_n^{[b]}(k)$ gives the probability that an agent adopts opinion~$1$ in the next round when operating the corresponding protocol.
Specifically, the process is obtained by performing the following two steps for every agent~$i$ simultaneously, in every round~$t$:
\begin{enumerate}
    \item A vector~$S_t^{(i)} \in I^\ell$ of size~$\ell$ is sampled uniformly at random (u.a.r.) (the same agent may appear several times in $S_t^{(i)}$, and~$i$ may appear in it as well).
    \item Writing~$k_t^{(i)}$ to denote the number of agents with opinion $1$ in $S_t^{(i)}$, Agent~$i$ updates its opinion according to
    \begin{equation*}
        X_{t+1}^{(i)} \leftarrow 1 \text{ with probability } g_n^{\left[X_t^{(i)}\right]}(k_t^{(i)}), \quad 0 \text{ otherwise}.
    \end{equation*}
\end{enumerate}
For the sake of clarity, we note a few important consequences of this definition.
\begin{itemize}
    \item Non-source agents do not know where the opinions that they observe come from. In particular, they do not know if $S_t^{(i)}$ contains the source.
    \item Non-source agents do not have identifiers, or in other words, all of them must run exactly the same update rule. They are also not aware of the round number (indices are used for analysis purposes only).
    \item However, non-source agents are aware of their current opinion, as well as the exact value of~$n$.
    \item Besides their opinion, non-source agents have no memory, in the sense that their behavior cannot depend on any information from previous rounds.
\end{itemize}
Since agents have no memory besides their opinion, and no identifiers, the configuration of the system in round~$t$ can be described simply by a pair~$(z,X_t)$, where $z \in \{0,1\}$ denotes the correct opinion, and~$X_t \in \{0,\ldots,n\}$ denotes the number of agents with opinion~$1$.
For a given~$n \in \bbN$, and an initial configuration~$C = (z,X_0)$, we define the convergence time of protocol~$\prot$ as the first round for which all agents have adopted the correct opinion and remain with it forever, that is:
\begin{equation*}
	\tau_n(\prot,C) := \inf \{t \geq 0, \text{ for every } s \geq t, X_s = n \cdot z \}.
\end{equation*}
Given a sequence of events~$\{A_n\}_{n\in \bbN}$, we say that ``$A_n$ happens with high probability (w.h.p.)'' if $\Pr (A_n) = 1-1/n^{\Omega(1)}$.

\subsection{Our Results} \label{sec:our_results}

We show that, when the sample size is bounded, any protocol that does not have access to memory needs almost-linear time to solve the bit-dissemination problem.
\begin{theorem} \label{thm:main}
    Assume that the sample size~$\ell$ is constant.
    For every protocol~$\prot$, there exists a sequence of initial configuration~$C_n$ such that for every~$\varepsilon > 0$,
    the convergence time of $\prot$ is greater than $n^{1-\varepsilon}$ w.h.p.: 
    \begin{equation*}
    	\Pr \big[ \tau_n(\prot,C_n) < n^{1-\varepsilon} \big] =  \frac{1}{n^{\Omega(1)}}.
    \end{equation*}
\end{theorem}
To the best of our knowledge, this is the first non-trivial lower-bound for this problem in the parallel setting.
The proof of \Cref{thm:main} is presented in \Cref{sec:main_proof}, and uses a general result on Markov chains, described in \Cref{sec:markov_chain_lemma}, as a black box.
It consists in studying a kind of ``characteristic'' function, defined as
\begin{equation} \label{eq:Fn_def}
    F_n(p) := -p + \sum_{k=0}^\ell \binom{\ell}{k} p^k (1-p)^{\ell - k} \pa{ p \, g_n^{[1]}(k) + (1-p) \, g_n^{[0]}(k) }.
\end{equation}
The sum in \Cref{eq:Fn_def} corresponds to the probability that a non-source agent, taken uniformly at random, adopts opinion 1, given that the current proportion of agents with opinion~$1$ is~$p$.
Informally, $F_n(p)$ measures the ``bias'' of a protocol~$\prot$ towards opinion~$1$, or in other words:
\begin{equation*} 
    \E \pa{ \frac{X_{t+1}}{n} \mid X_t = x_t} \approx \frac{x_t}{n} + F_n\pa{\frac{x_t}{n}}
\end{equation*}
(see \Cref{lem:expectation_bounds} for a more accurate statement).
As a consequence, the sign of the function at~$p$ provides information on the trend of the dynamics when the proportion of agents with opinion 1 is~$p$; moreover, roots correspond to fixed points of the dynamics (either stable or unstable).

Similar functions have already been defined in the literature, where they are typically used to obtain two types of results.
On the one hand, they can be used to identify phase transitions~\cite{cruciani_phase_2021,cruciani_phase_2021a}.
For example, when such characteristic function depends on a parameter~$\alpha$, there might be a critical value $\alpha^\star$ at which a new root appears. In that case, the behavior of the dynamics below and above~$\alpha^\star$ can be significantly different.
On the other hand, if a characteristic function has a constant sign over a large interval (and under some additional conditions), the dynamics cannot easily travel in the opposite direction, which can be exploited to obtain lower bounds on the convergence time~\cite{cruciani_phase_2021}.

In this paper, we use function~$F_n$ yet in another way. Specifically, we leverage the fact that a bound on the sample size~$\ell$ implies a bound on the degree of~$F_n$, and therefore a bound on the number of roots within the interval~$[0,1]$. Then, we consider a well-chosen interval of constant length between two roots of~$F_n$, and employ the aforementioned argument to obtain a lower bound:
if~$F_n$ is negative on this interval, i.e., $\prot$ tends to make the proportion of $1$-opinions decrease, we show that the process will be slow to reach consensus every time the correct opinion is~$1$.
Conversely, if it is positive, we show that fast convergence fails whenever the correct opinion is~$0$.

In terms of the dependency on~$n$, we show that our lower bound is nearly tight (up to a sub-polynomial factor) by adapting a well-known result to our setting. Its proof does not introduce any novel argument and is deferred to \Cref{sec:voter_upper}.
\begin{theorem} \label{thm:secondary}
    Consider the Voter dynamics~$\prot^\textrm{voter}$, with sample size~$\ell = 1$. For every sequence of initial configuration~$C_n$, 
    the convergence time of $\prot^\textrm{voter}$ is less than $2n \log n$ w.h.p.:
    \begin{equation*}
    	\Pr \big[ \tau_n(\prot,C_n) \leq 2n \log n \big] \geq 1 - \frac{1}{n^2} .
    \end{equation*}
\end{theorem}
When it comes to parameter~$\ell$, a gap remains between our lower bound and the upper bound in~\cite{becchetti_minority_2024}, where it is shown that the minority dynamics solves the problem in~$O(\log^2 n)$ rounds w.h.p. when~$\ell$ is at least~$\Omega(\sqrt{n \log n})$.

Unfortunately, we believe that our techniques cannot be used to extend the lower bound to a higher value of~$\ell$. 
Indeed, if $\ell = \Omega(\log n)$, it is already possible for a protocol to converge in just one round w.h.p. from configurations that are arbitrarily far away from the consensus. This observation destroys any hope of restricting the analysis to a small interval of the configuration space. In contrast, this phenomenon does not happen w.h.p. in our setting (see \Cref{lemma:no_jump_to_consensus}).
However, we have no good reason to think that $\Theta(\sqrt{n \log n})$ is the smallest value of~$\ell$ allowing for an efficient protocol (such as the minority dynamics) -- this is left as an open problem.

\subsection{Other Related Works}

The bit-dissemination problem was also studied under the assumption that agents can use a moderate amount of memory.
An efficient protocol is identified in \cite{korman_early_2022}, and achieves consensus in $O(\polylog~n)$ parallel rounds with high probability. 
It relies on agents being able to memorize $\log \log n$ bits of information from one round to the next and requires a sample size logarithmic in~$n$.
Other candidates are mentioned in~\cite{boczkowski_minimizing_2017} but are not analysed.
The authors of \cite{dudek_universal_2018} show that the problem can be solved in the context of population protocols with a memory of only constant size. Importantly however, population protocols do not fit the framework of passive communications.
Specifically, interaction rules in this model depend on the exact states of the agents, and not just on their binary opinion.

The bit-dissemination problem is a specific case of the \textit{majority bit-dissemination} problem, introduced in~\cite{boczkowski_minimizing_2017} and also addressed in~\cite{dudek_universal_2018}. In this variant, the number of source agents is arbitrarily large, and they may have conflicting \textit{preferences}.
The opinions of sources must not necessarily be in line with their preferences, and they can participate to the protocol in the same way as regular agents.
The correct opinion is defined as the most widespread preference among sources.
On the one hand, an efficient solution can be derived from the results in~\cite{dudek_universal_2018,bastide_selfstabilizing_2021}, but require active communications, and relies on memory.
On the other hand, the authors of~\cite{korman_early_2022} show that the majority bit-dissemination problem is impossible with passive communications.

More generally, many works within the opinion dynamics literature investigate the influence of the presence of ``stubborn'' or ``biased'' agents on the behaviour of the system.
Typically, these works focus on a single arbitrary process, mainly the Voter dynamics~\cite{sood_voter_2005,yildiz_binary_2013a,fudolig_analytic_2014,mukhopadhyay_binary_2016,moeinifar_zealots_2021}, and to the best of our knowledge, do not establish general lower bounds.
In contrast, our goal is to better understand the difficulty of spreading information as an algorithmic problem; therefore, we do not want to rule out any imaginable protocol within the constraints of our setting.
Furthermore, they often investigate different questions, such as the impact of the number of sources or they position in the network on the convergence time, or assume that sources may have conflicting opinions.

\section{Preliminaries} \label{sec:preliminaries}

In this section, we make a few general observations that we will use later in our analysis. We assume that~$\ell$ is a constant w.r.t.~$n$.

Conditioning on Agent~$i$ sampling exactly~$k$ times the opinion ``1'', for every~$k \in \{0,\ldots,\ell\}$, we obtain that
\begin{equation} \label{eq:general_proba}
    \Pr \pa{X_{t+1}^{(i)} = 1 \mid X_t = x_t,~ X_t^{(i)} = b} = \sum_{k=0}^\ell \binom{\ell}{k} \pa{\frac{x_t}{n}}^k \pa{1-\frac{x_t}{n}}^{\ell - k} g_n^{[b]}(k).
\end{equation}

After convergence has happened, i.e., $X_t \in \{0,n\}$, this probability must be equal to~$0$ or~$1$ respectively so that a consensus is maintained. This imposes a constraint on any protocol attempting to solve the bit-dissemination problem, which can be formalized as follows.

\begin{proposition} \label{prop:necessary}
	Any protocol~$\prot$ solving the bit-dissemination problem must satisfy $g_n^{[0]}(0) = 0$ and $g_n^{[1]}(\ell) = 1$.
\end{proposition}
\begin{proof}
	Consider the case that the correct opinion $z=0$.
	If $X_t=0$ for some round~$t$, then each agent has opinion~$0$, and receives exactly~$0$ samples equal to~'$1$'. Following the protocol, all agents adopt opinion $1$ in the next round w.p. $g_n^{[0]}(0)$ independently of each-other. If $g_n^{[0]}(0)>0$, then
	\begin{equation*}
		\Pr \pa{ X_{t+1} = 0 \mid X_t = 0} = \prod_{i\geq 2} \Pr \pa{ X_{t+1}^{(i)} = 0 \mid X_t = 0} = \pa{1 - g_n^{[0]}(0)}^{n-1} < 1.
	\end{equation*}
	Therefore, $\inf \{t \geq 0, \text{ for every } s \geq t, X_s = 0 \} = +\infty$ almost surely, and the protocol cannot solve the bit-dissemination problem in the sense of \Cref{sec:problem_definition}.
	By symmetry, we obtain the other statement about $g_n^{[1]}(\ell)$, which concludes the proof of \Cref{prop:necessary}.
\end{proof}

Accordingly, we will always assume that~$g_n^{[0]}(0) = 0$ and $g_n^{[1]}(\ell) = 1$.
Using this assumption, we can show a general upper bound on the fraction of agents with opinion~$0$ that can change opinion in a single time step.

\begin{proposition} \label{lemma:no_jump_to_consensus}
Let $c \in (0,1)$ and consider a protocol~$\prot$ solving the bit-dissemination problem. There is a constant $y = y(c,\ell) \in(c,1)$ s.t. for every~$n$ large enough, and~$x_t \leq c \, n$, 
\begin{equation*}
    \Pr \pa{ X_{t+1} \leq y \, n \mid X_t=x_t } \geq 1 - \exp\pa{ - 2 \, n^{-1/2}}.
\end{equation*}
\end{proposition}
\begin{proof}
    Let~$t \in \bbN$, and~$x_t \leq cn$. By \Cref{eq:general_proba}, and since we assumed $g_n^{[0]}(0) = 0$,
    \begin{align*}
        \Pr \pa{X_{t+1}^{(i)} = 0 \mid X_t = x_t, X_t^{(i)} = 0} &= 1 - \sum_{k=0}^\ell \binom{\ell}{k} \pa{\frac{x_t}{n}}^k \pa{1-\frac{x_t}{n}}^{\ell - k} g_n^{[0]}(k) \\
        &= \sum_{k=0}^\ell \binom{\ell}{k} \pa{\frac{x_t}{n}}^k \pa{1-\frac{x_t}{n}}^{\ell - k} \pa{ 1-g_n^{[0]}(k) } \\
        &\geq \pa{ 1- \frac{x_t}{n} }^\ell \geq (1-c)^\ell.
    \end{align*}
    Let $Y$ be the number of agents with opinion~$0$ in round $t$, that keep opinion~$0$ in round~$t+1$ (conditioning on $X_t = x_t$).
    By assumption, $n-X_t \geq (1-c)n$, so the last equation implies the following domination\footnote{Given two real-valued random variables $X$ and $Y$, we say that~$X$ is \textit{stochastically dominated} by $Y$, and write $X \preceq Y$, if for every~$x \in \bbR$, $\Pr(X > x) \leq \Pr(Y > x)$.}:
    \begin{equation*}
        Y \succeq \binomial \pa{ (1-c)n , (1-c)^\ell } := Z.
    \end{equation*}
    Let~$a = a(c,\ell) := (1-c)^{\ell+1}$, so that $\E\pa{Z} = an$.
    Let $a' := a - n^{-1/4}$. For~$n$ large enough, we have $a' > a/2$.
    By \nameref{thm:additive_chernoff_bound}, we have
    \begin{align*}
        \Pr \pa{Z \leq \frac{an}{2}} \leq \Pr \pa{Z \leq a'n} = \Pr \pa{Z \leq \E(Z) -n^{3/4}} \leq \exp\pa{ -2n^{1/2} }.
    \end{align*}
    Finally, setting~$y = y(c,\ell) := \max( 1-a/2,~ c)$ (so that~$y \in (c,1)$), we obtain
    \begin{equation*}
        \Pr \pa{ X_{t+1} \geq y \, n \mid X_t=x_t } \leq \Pr \pa{ Y \leq \frac{an}{2} } \leq \Pr \pa{ Z \leq \frac{an}{2} } \leq \exp\pa{ -2 \, n^{1/2} },
    \end{equation*}
    which concludes the proof of \Cref{lemma:no_jump_to_consensus}.
\end{proof}

Finally, the following proposition justifies the informal claim made in \Cref{sec:our_results} according to which the function $F_n$, defined in \Cref{eq:Fn_def}, represents the ``bias'' of the corresponding protocol towards opinion~$1$.

\begin{proposition} \label{lem:expectation_bounds}
    For every protocol~$\prot$, and every $x_t \in [n]$,
    \begin{align} 
        \E \pa{ X_{t+1} \mid X_t = x_t} &\leq x_t + n \, F_n\pa{\frac{x_t}{n}} +1, \label{eq:expectation_bound_z1n} \\
        \E \pa{ X_{t+1} \mid X_t = x_t} &\geq x_t + n \, F_n\pa{\frac{x_t}{n}} -1. \label{eq:expectation_bound_z0n}
    \end{align}
\end{proposition}
\begin{proof}
    For a non-source agent~$i \in I\setminus\{1\}$, an opinion~$b \in \{0,1\}$, and any $p \in [0,1]$, let
    \begin{equation*}
        P_b := \Pr \pa{X_{t+1}^{(i)} = 1 \mid X_t = np,~ X_t^{(i)} = b} = \sum_{k=0}^\ell \binom{\ell}{k} p^k (1-p)^{\ell - k} g_n^{[b]}(k),
    \end{equation*}
    where the second equality is a restatement of \Cref{eq:general_proba}.
    Note that by definition of~$F_n$,
    \begin{equation} \label{eq:F_restatement}
        F_n(p) = p \, P_1 + (1-p)P_0 - p.
    \end{equation}
    Denoting by~$z$ the correct opinion, one can check that there are $X_t-z$ non-source agents with opinion 1 in round $t$, and $n-X_t-(1-z)$ non-source agents with opinion 0. Hence,
    \begin{align*}
        \E \pa{ X_{t+1} \mid X_t = np } &= z + (np-z) P_1 + \big(n - np - (1-z) \big)P_0 \\
        &= n \big(pP_1 + (1-p)P_0 \big) + z(1-P_1) - (1-z)P_0 \\
        &= np + n \, F_n(p) + z(1-P_1) - (1-z)P_0. & \text{(by \Cref{eq:F_restatement})}
    \end{align*}
    Note that for any source opinion $z \in \{0,1\}$, since $P_0,P_1 \in [0,1]$, we have
    \begin{equation*}
        -1 \leq z(1-P_1) - (1-z)P_0 \leq +1,
    \end{equation*}
    from which \Cref{eq:expectation_bound_z1n,eq:expectation_bound_z0n} follow by taking $p := x_t/n$.
\end{proof}

\section{An Intermediate Result on Markov Chains} \label{sec:markov_chain_lemma}

In this section, we present a result that we will use later as a black box (with~$a_1,a_2,a_3 \in [0,1]$).
Informally, the theorem says that if a Markov chain is a super-martingale over an interval of values, and given that it cannot skip the interval entirely, then the time required to cross the interval is at least the time needed by a martingale to escape it.

\begin{theorem} \label{lem:main}
    Let $\{X_t\}_{t \in \bbN}$ be a Markov chain on $\mathbb{Z}$ and $\varepsilon > 0$.
    If there are $a_1<a_2<a_3 \in \bbR$ s.t.
\begin{itemize}
    \item[(i)] for every $x_t \in \{ \lceil a_1 \, n \rceil,..., \lfloor a_3 \, n \rfloor \}$, $\E(X_{t+1} \mid X_t = x_t) \leq x_t+1$,
    \item[(ii)] for every $x_t < a_1 \, n$,
    $\Pr (X_{t+1} > a_2 \, n \mid X_t = x_t) = \exp\pa{-n^{\Omega(1)}}$,
    \item[(iii)] $\Pr(|X_{t+1}-\E\pa{X_{t+1} \mid X_t}| > n^{1/2 + \varepsilon/4}) < 2\exp \pa{ -2n^{\varepsilon/2}}$,
\end{itemize}
then for $X_0=\frac{a_2+a_3}{2} \cdot n$ and $n$ large enough, we have w.h.p.
\begin{equation*}
    \inf \{t \in \bbN, X_t \geq a_3 \, n\} \geq n^{1-\varepsilon}.
\end{equation*}
\end{theorem}
\begin{proof} 
Let~$T = n^{1-\epsilon}$.
Let $Y_t:=X_t-t$. We will consider the Doob decomposition of $Y_t$: for every $t\geq 1$, let
\begin{align*}
    A_0 := 0 &\quad \text{and for all } t>0, \quad A_t :=\sum_{k=1}^t \big[ \E\pa{Y_k \mid Y_{k-1}} -Y_{k-1} \big], \\
    M_0 := Y_0 &\quad \text{and for all } t>0, \quad M_t := Y_0+\sum_{k=1}^t \big[ Y_k - \E\pa{ Y_k \mid Y_{k-1} } \big].
\end{align*}
With this definition, one can check that $Y_t = M_t + A_t$, and that $\{M_t\}_{t \in \bbN}$ is a martingale.
The main ideas of the proof are depicted on \Cref{fig:mc_proof_sketch}.
\begin{figure} [htbp] 
	\centering
	\includegraphics[width=0.75\linewidth]{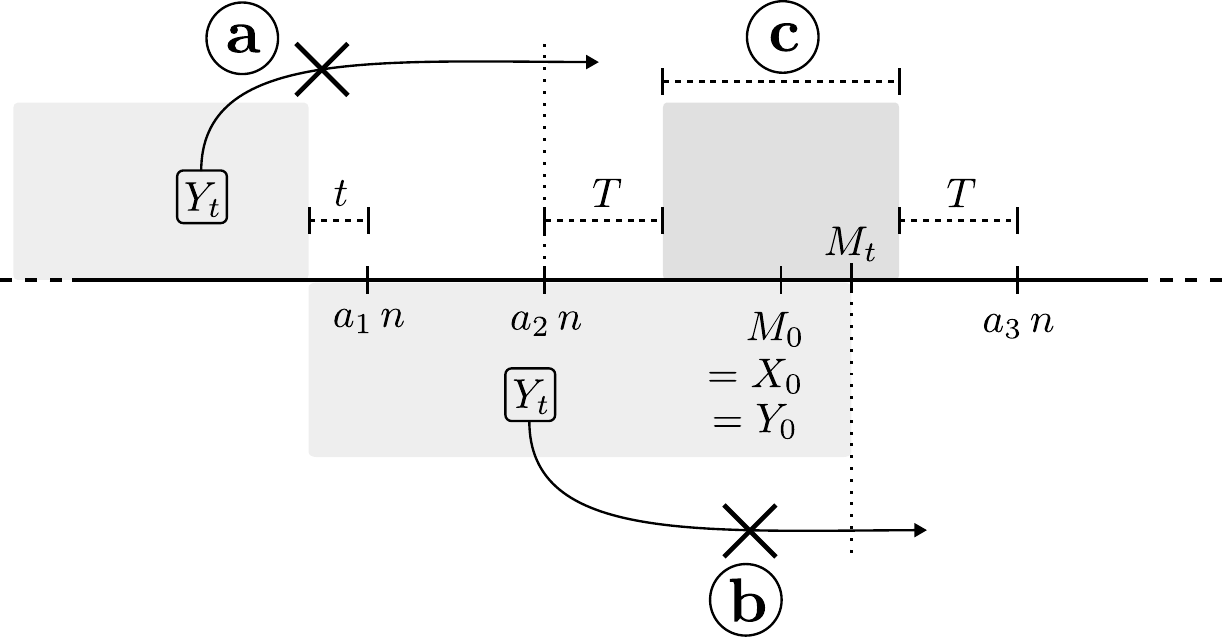}
	\caption{Sketch of the proof of \Cref{lem:main}.
	\textbf{(a)} By assumption (ii), with high probability, $Y_t$ cannot jump from below $a_1 \,n - t$ to above $a_2 \, n -t$ in a single step, let alone~$a_2 \, n$.
	\textbf{(b)} In \Cref{claim:one_step_domination}, we use assumption (i) and the properties of the Doob's decomposition to show that, if $a_1 \, n - t \leq Y_t \leq M_t$, then $Y_t$ cannot jump above $M_t$ in a single step.
	\textbf{(c)} In \Cref{claim:trapping_martingale}, we use the Azuma-Hoeffding inequality to show that $M_t$ remains in the interval $[a_2\,n+T,a_3\,n-T]$ for at least $T$ rounds w.h.p.
	Overall, (a) (b) and (c) implies that $Y_t$ must remain below $a_3 \,  n - T$ for at least $T$ rounds w.h.p., yielding the desired conclusion.
	}
    \label{fig:mc_proof_sketch}
\end{figure}

First, we show that by construction and Assumption (i), $Y_t$ can never ``jump over'' $M_t$ in one round, as long as it starts from the interval~$\{a_1 n-t, \ldots, a_3 n-t\}$.
\begin{claim} \label{claim:one_step_domination}
    For every~$t \in \bbN$,
    \begin{equation*} \label{domination:Y}
        M_t \geq Y_t \text{ and } Y_t \in \{a_1 n-t, \ldots, a_3 n-t\}  \implies M_{t+1} \geq Y_{t+1}.
    \end{equation*}
\end{claim}
\begin{proof}
    Let $y_t \in \{a_1 n-t, \ldots, a_3 n-t\}$ and $m_t \geq y_t$, and consider the event
    \begin{equation*}
        \event := \left\{ Y_t = y_t \medcap M_t = m_t \right\}. 
    \end{equation*}
    We have
    \begin{align*}
		\E \pa{Y_{t+1} \mid Y_t = y_t} &= \E \pa{X_{t+1} - (t+1) \mid X_t = y_t + t} \\
		&= \E \pa{X_{t+1} \mid X_t = y_t + t} - (t+1) \\
		&\leq (y_t + t + 1) - (t+1) & \text{(by (i))}\\
		&= y_t.
	\end{align*}
    Therefore,
    \begin{equation*}
        \pa{A_{t+1} - A_t \mid \event} = \bigg( \E \pa{Y_{t+1} \mid Y_t} - Y_t \mid \event \bigg) = \E \pa{Y_{t+1} \mid Y_t = y_t} - y_t \leq 0.
    \end{equation*}
    Since $m_t \geq y_t$, this implies
    \begin{equation*}
        \pa{A_{t+1} \mid \event} \leq \pa{A_{t} \mid \event} = \pa{Y_t - M_t \mid \event} = y_t - m_t \leq 0,
    \end{equation*}
    and thus,
    \begin{equation*}
        \pa{Y_{t+1} \mid \event} = \pa{M_{t+1} + A_{t+1} \mid \event} \leq \pa{M_{t+1} \mid \event},
    \end{equation*}
    which concludes the proof of \Cref{claim:one_step_domination}.
\end{proof}

Now, using Azuma's inequality and Assumptions (ii) and (iii), we establish high probability bounds on the martingale~$M_t$.
\begin{claim} \label{claim:trapping_martingale}
    With high probability, for every~$t \leq T$, $a_2 \, n + T < M_t < a_3 \, n - T$.
\end{claim}
\begin{proof}
    By construction, and since~$X_t$ and $Y_t$ only differ by a deterministic quantity,
    \begin{equation*}
        M_{t+1} - M_t = Y_{t+1}- \E \pa{Y_{t+1} \mid Y_t} = X_{t+1}- \E\pa{X_{t+1} \mid X_t}.
    \end{equation*}
    By assumption (iii) in the statement, this implies
    \begin{equation*}
        \Pr\pa{ |M_{t+1} - M_t| > n^{\frac{1}{2} + \frac{\varepsilon}{4}} }\leq 2\exp\pa{ -2 \, n^{\frac{\varepsilon}{2}} }.
    \end{equation*}
    By the union bound,
	\begin{equation*} \label{eq:hypothesis_for_Azuma_inequality}
		\Pr \pa{ \exists s \leq t, |M_{s+1} - M_s| > n^{\frac{1}{2} + \frac{\varepsilon}{4}} } \leq 2 \, t \exp \pa{ - 2 \, n^{\frac{\varepsilon}{2}} }.
	\end{equation*}
    Let~$\alpha = (a_3 - a_2)/4$, so that
    \begin{equation} \label{eq:def_alpha}
        M_0 + \alpha \, n = X_0 + \alpha \, n = \frac{a_2+a_3}{2} \, n + \alpha \, n = a_3 \, n - \alpha \, n, \quad \text{ and } M_0 - \alpha \, n = a_2 \, n + \alpha \, n.
    \end{equation}
    By the \nameref{thm:generalized_azuma_hoeffding} applied to $\{M_t\}_{t \in \bbN}$, we then have for every~$t \leq T$,
	\begin{align*}
		\Pr \pa{ |M_t - M_0| > \alpha n} &\leq 2\exp \pa{-\frac{\alpha^2 n^2}{2 \, t \, n^{1+\frac{\varepsilon}{2} }}} + 2 \, t \exp \pa{ - 2 \, n^{\frac{\varepsilon}{2}} } & \\
        &\leq 2 \exp \pa{-\frac{\alpha^2 n^2}{2 \, T \, n^{1+\frac{\varepsilon}{2} }}} + 2 \, T \exp \pa{ - 2 \, n^{\frac{\varepsilon}{2}} } & \text{(since $t \leq T$)} \\
		&= 2 \exp \pa{-\frac{\alpha^2}{2} \cdot n^{\frac{\varepsilon}{2}}} + 2 \, T \exp \pa{ - 2 \, n^{\frac{\varepsilon}{2}} }. & \text{(since $T = n^{1-\varepsilon}$)}
	\end{align*}
    By the union bound,
    \begin{equation} \label{eq:final_bound}
        \Pr \pa{ \exists t \leq T, |M_t - M_0| > \alpha n } \leq 2 T \exp \pa{-\frac{\alpha^2}{2} \cdot n^{\frac{\varepsilon}{2}}} + 2 \, T^2 \exp \pa{ - 2 \, n^{\frac{\varepsilon}{2}} } = o(n^{-2}).
    \end{equation}
    Finally, note that for $n$ large enough, $T = n^{1-\epsilon} < \alpha \, n$, and hence,
    \begin{align*}
        &\Pr\pa{ \exists t \leq T, M_t \notin \{ a_2 \, n + T, \ldots, a_3 \, n - T \} } \\
        &\leq \Pr\pa{ \exists t \leq T, M_t \notin \{ a_2 \, n + \alpha \, n, \ldots, a_3 \, n - \alpha \, n \} } \\
        &= \Pr\pa{ \exists t \leq T, |M_t - M_0| > \alpha \, n } & \text{(by \Cref{eq:def_alpha})} \\
        &= o(n^{-2}), & \text{(by \Cref{eq:final_bound})}
    \end{align*}
    which concludes the proof of \Cref{claim:trapping_martingale}.
\end{proof}

Next, we use \Cref{claim:one_step_domination} and \Cref{claim:trapping_martingale} and Assumption (ii) of the theorem to show that $Y_t$ can never jump over~$M_t$, with high probability.
\begin{claim} \label{claim:domination}
    With high probability, for every $t \leq T$, $M_t \geq Y_t$.
\end{claim}
\begin{proof}
    We will be conditioning on the two following events:
    \begin{align*}
        \event_1 &:= \left\{ \forall t \leq T, \quad Y_t \leq a_1\,n - t \implies Y_{t+1} \leq a_2\,n - t \right\}, \\
        \event_2 &:= \left\{ \forall t \leq T, \quad a_2 \, n + T < M_t< a_3\,n - T \right\}.
    \end{align*}
    Note that $\event_2$ happen w.h.p. as a consequence of \Cref{claim:trapping_martingale}. Moreover,
    \begin{align*}
        \Pr(\event_1) 
        &= 1 - \Pr\pa{ \bigcup\limits_{t=0}^{T} \{Y_t \leq a_1\,n - t \} \medcap \{ Y_{t+1} > a_2\,n - t \} }\\
        &= 1 - \Pr\pa{ \bigcup\limits_{t=0}^{T} \{X_t \leq a_1\,n \} \medcap \{ X_{t+1} > a_2\,n \} }\\
        &\geq 1 - \sum_{t=0}^T \Pr \pa{ X_t \leq a_1\,n  ~\medcap~ X_{t+1} > a_2\,n } & \text{(by the union bound)} \\
        &\geq 1 - \sum_{t=0}^T \Pr \pa{ X_{t+1} > a_2\,n \mid X_t \leq a_1\,n } & \\
        &\geq 1-T\exp\pa{-n^{\Omega(1)}} & \text{(by assumption (ii))}\\
        &\geq 1-o(n^{-2}). &\text{(since $T = n^{1-\varepsilon}$)}
    \end{align*}
    Hence, $\event_1$ also happens w.h.p., and so does $\event_1 \medcap \event_2$.
    To conclude the proof, we will show that
    \begin{equation} \label{eq:induction_target}
        \event_1 \medcap \event_2 \implies \forall t \leq T, M_t \geq Y_t.
    \end{equation}
    We will proceed by induction on~$t$. By definition, we have $M_0 = Y_0$.
    Now, let~$t < T$, and consider the case that $\event_1$ and $\event_2$ hold, and that $M_t \geq Y_t$.
    \begin{itemize}
        \item If $Y_t < a_1 \, n - t$, we have
        \begin{equation*}
            Y_{t+1} \underset{(\event_1)}{\leq} a_2 \, n - t \leq a_2 \, n + T \underset{(\event_2)}{<} M_{t+1}.
        \end{equation*}

        \item Otherwise, by induction hypothesis and $\event_2$, we have $a_1 \, n - t \leq Y_t \leq M_t < a_3 n - T < a_3 n - t$. In this case, $M_{t+1} \geq Y_{t+1}$ follows as a consequence of \Cref{claim:one_step_domination}.
    \end{itemize}
    By induction, we deduce that \Cref{eq:induction_target} holds, which concludes the proof of \Cref{claim:domination}.    
\end{proof}

Finally, we are ready to conclude.
By \Cref{claim:domination}, with high probability, for every~$t\leq T$, $X_t \leq M_t + t$.
Therefore,
\begin{equation*}
    \inf \{t \in \bbN, X_t > a_3 n\} \geq \inf \{t \in \bbN, M_t > a_3 n-t \} \quad \text{w.h.p.}
\end{equation*}
Moreover, by \Cref{claim:trapping_martingale},
\begin{equation*}
    \inf \{t \in \bbN, M_t > a_3 n-t \} > T \quad \text{w.h.p.},
\end{equation*}
which gives the desired result.
\end{proof}

By symmetry, the following result can be deduced directly from \Cref{lem:main}.
\begin{corollary} \label{cor:main_lemma_reversed}
    Let $\{X_t\}$ be a Markov chain on $\mathbb{Z}$ and $\varepsilon > 0$.
    If there are $a_1<a_2<a_3 \in \bbR$ s.t.
\begin{itemize}
    \item[(i)] for every $x_t \in \{ \lceil a_1 \, n \rceil,..., \lfloor a_3 \, n \rfloor \}$, $\E(X_{t+1} \mid X_t = x_t) \geq x_t-1$,
    \item[(ii)] for every $x_t > a_3 \, n$, $\Pr (X_{t+1} < a_2 \, n \mid X_t = x_t)=\exp(-n^{\Omega(1)})$,
    \item[(iii)] $\Pr(|X_{t+1}-\E\pa{X_{t+1} \mid X_t}| > n^{1/2+\varepsilon/4}) < 2\exp \pa{ -2n^{\varepsilon/2}}$,
\end{itemize}
then for $X_0=\frac{a_1+a_2}{2} \cdot n$ and for $n$ large enough, we have w.h.p.
\begin{equation*}
    \inf \{t \in \bbN, X_t \leq a_1 \, n\} \geq n^{1-\varepsilon}.
\end{equation*}
\end{corollary}
\begin{proof}
    It is easy to check that if the assumptions of \Cref{cor:main_lemma_reversed} hold for $\{X_t\}_{t \in \bbN}$ and some constants $a_1 < a_2 < a_3$, then the assumptions of \Cref{lem:main} hold w.r.t. $\{-X_t\}_{t \in \bbN}$ and $-a_3 < -a_2 < -a_1$ respectively -- and the conclusion follows.
\end{proof}

\section{The Main Proof} \label{sec:main_proof}

\subsection{The Voter Dynamics} \label{sec:voter_model}

First, we focus on the Voter dynamics (\Cref{prot:voter})
and show that it satisfies the lower bound stated in \Cref{thm:main}, which we will prove in its full generality in \Cref{sec:general_lower_bound}.
We start by observing that for~$g^\textrm{voter}$ defined in \Cref{eq:def_voter}, and by definition in \Cref{eq:Fn_def},
\begin{equation*}
    F_n^\textrm{voter}(p) = -p + \sum_{k=0}^\ell \binom{\ell}{k} p^k (1-p)^{\ell - k} \cdot \frac{k}{\ell} = -p + p = 0.
\end{equation*}
Then, we conclude by applying the following result. Due to space constraints, its proof is deferred to \Cref{sec:missing_proofs}.

\begin{lemma} \label{thm:voter_lower}
    Consider a protocol~$\prot$ satisfying $F_n = 0$ for every~$n$ large enough.
	There exists a sequence of configurations~$\{C_n\}$ such that for every~$\epsilon > 0$, with high probability,
	\begin{equation*}
		\tau_n(\prot,C_n) > n^{1-\varepsilon}.
	\end{equation*}
\end{lemma}

\subsection{General Case} \label{sec:general_lower_bound}

Our main result (\Cref{thm:main}) will follow as a consequence of \Cref{thm:general_lower_bound} below.
\begin{theorem} \label{thm:general_lower_bound}
    For every~$\varepsilon > 0$ and every protocol~$\prot$, there exists an infinite set~$\support \subseteq \bbN$ and a sequence of configurations~$\{C_n\}$ such that for every~$n \in \support$,
    \begin{equation*}
        \Pr \pa{\tau_n(\prot,C_n) > n^{1-\varepsilon}} \geq 1 - \frac{1}{n^{\Omega(1)}}.
    \end{equation*}
    In other words, the convergence time of~$\prot$ restricted to $\support$ is greater than $n^{1-\varepsilon}$ w.h.p.
\end{theorem}
\begin{proof}
Recall the definition of~$F_n$ in \Cref{eq:Fn_def}.
If there is $N \in \bbN$ s.t. for every~$n \geq N$, $F_n = 0$, then we can conclude by applying \Cref{thm:voter_lower}.
Otherwise, there is an infinite set~$\support_0 \subseteq \bbN$ s.t. for every~$n \in \support_0$, $F_n \neq 0$, which we will assume from now on.

By definition in \Cref{eq:Fn_def}, $F_n$ is a polynomial of degree at most~$\ell+1$.
For~$n \in \support_0$, let $d_n$ be the number of roots of~$F_n$ in the interval~$[0,1]$ (counted with multiplicity). 
By \Cref{prop:necessary}, $g_n^{[0]}(0) = 0$ and $g_n^{[1]}(\ell) = 1$, so $F_n(0) = F_n(1) = 0$, and thus $d_n \in \{ 2, \ldots, \ell+1 \}$.
Since $d_n$ can only adopt finitely many values, there exists $d \in \{ 2, \ldots, \ell+1\}$ and an infinite set~$\support_1 \subseteq \support_0$ s.t. for every~$n \in \support_1$, $d_n = d$.

For~$n \in \support_1$, let~$0 = r_n^{(1)} \leq \ldots \leq r_n^{(d)} = 1$ be the roots of $F_n$ within the interval~$[0,1]$, with multiplicity, in increasing order.
The sequence $\big\{ (r_n^{(1)}, \ldots , r_n^{(d)}) \big\}_{n \in \support_1}$ is bounded in $\bbR^d$ by definition. Hence, by the Bolzano-Weierstrass theorem, there exists a converging sub-sequence of $\big\{ ( r_n^{(1)}, \ldots, r_n^{(d)}) \big\}_{n \in \support_1}$; i.e., there are $0 = \rinf^{(1)} \leq \ldots \leq \rinf^{(d)} = 1$ together with an infinite set $\support_2 \subseteq \support_1$ s.t. for every $k \in [d]$,
\begin{equation} \label{eq:root_limit}
    \lim_{\substack{n \rightarrow +\infty \\ n \in \support_2}} r_{n}^{(k)} = \rinf^{(k)}.
\end{equation}
Let $k_0 := \min \{ k \in [d], \rinf^{(k)} = 1 \}$.
Note that $k_0 \geq 2$ (since $\rinf^{(1)} = 0$) and by definition, $\rinf^{(k_0-1)} < 1 = \rinf^{(k_0)}$.
Moreover, for every~$n \in\support_2$, $F_n$ is non-zero and has constant sign on~$(r_n^{(k_0-1)},r_n^{(k_0)})$.
Therefore, there exists an infinite set~$\support_3 \subseteq \support_2$ s.t.
\begin{enumerate}
    \item either $\forall n \in \support_3$, $F_n < 0$ on $(r_n^{(k_0-1)},r_n^{(k_0)})$,
    \item or $\forall n \in \support_3$, $F_n > 0$ on $(r_n^{(k_0-1)},r_n^{(k_0)})$.
\end{enumerate} 
In the remainder of the proof, we will analyse these two cases separately. The reader is strongly encouraged to consult \Cref{fig:proof_sketch_negative} and \ref{fig:proof_sketch_positive} respectively for a clearer understanding of the argument.

\begin{figure} [t]
    \centering
    \includegraphics[width=0.9\linewidth]{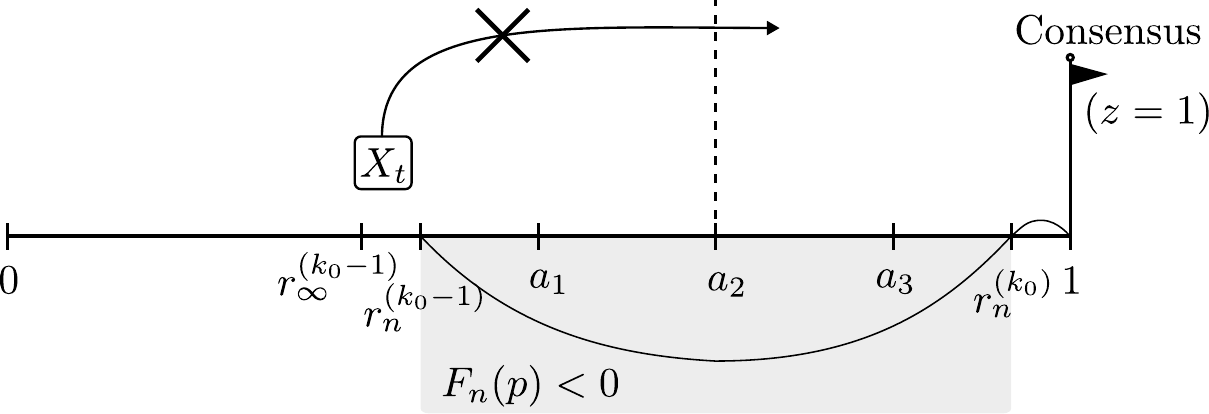}
    \caption{\textbf{Illustration of the arguments for Case 1.}
    We consider a configuration in which the correct opinion is~$1$.
    Constant~$a_1$ is fixed arbitrarily in the interval~$(\rinf^{(k_0-1)},1)$
    Then, $a_2$ is chosen according to \Cref{lemma:no_jump_to_consensus} to ensure that $X_t$ cannot jump from below $a_1 \, n$ to above $a_2 \, n$.
    Finally, $a_3$ is set anywhere in the interval~$(a_2,1]$.
    By assumption, $F_n < 0$ on $[a_1,a_3]$, and we can eventually apply \Cref{lem:main}.
    }
    \label{fig:proof_sketch_negative}
\end{figure}

\paragraph*{Case 1.}
Let $a_1 \in (\rinf^{(k_0-1)},1)$. Let~$a_2 = a_2(a_1,\ell) \in (a_1,1)$ given by \Cref{lemma:no_jump_to_consensus}, s.t. for~$n$ large enough,
\begin{equation*}
    \text{for every } x_t \leq a_1 \, n, \quad \Pr( X_{t+1} \leq a_2 \, n \mid X_t = x_t ) \geq 1 - \exp( - 2 \, n^{-1/2}).
\end{equation*}
Let~$a_3 \in (a_2,1)$.
By \Cref{eq:root_limit}, for~$n$ large enough, $r_n^{(k_0-1)} < a_1$ and $r_n^{(k_0)} > a_3$.
We now wish to use \Cref{lem:main}, with~$a_1,a_2,a_3$ as we just defined. Let us check that every assumption holds:
\begin{itemize}
    \item For every~$x_t \in \{a_1 \, n, \ldots, a_3 \, n \}$, we have~$x_t/n \in [a_1,a_3] \subset [r_n^{(k_0-1)},r_n^{(k_0)}]$,
    and so $F_n(x_t/n) < 0$ by assumption. Therefore, \Cref{eq:expectation_bound_z1n} gives
    \begin{equation}
        \E \pa{ X_{t+1} \mid X_t = x_t} \leq x_t + n \, F_n\pa{\frac{x_t}{n}} + 1 < x_t+1,
    \end{equation}
    so assumption (i) in the statement of \Cref{lem:main} holds.

    \item Assumption (ii) holds by \Cref{lemma:no_jump_to_consensus}.

    \item Finally, assumption (iii) follows from \nameref{thm:additive_chernoff_bound}: conditioning on $X_t$, $X_{t+1}$ is the sum of~$n$ Bernoulli random variables, then the result follow choosing $\delta=n^{1/2 + \varepsilon/4}$.
\end{itemize}
Assuming that the source has opinion~$z = 1$, we apply \Cref{lem:main}, which implies the existence of an initial configuration~$C_n$ s.t. the convergence time is bounded w.h.p.: 
\begin{equation*}
    \tau_n(g,C_n) =  \inf \{t \in \bbN, X_t = n\} \geq \inf \{t \in \bbN, X_t \geq a_3 \, n\} \geq n^{1-\varepsilon}.
\end{equation*}

\begin{figure} [t]
    \centering
    \includegraphics[width=0.9\linewidth]{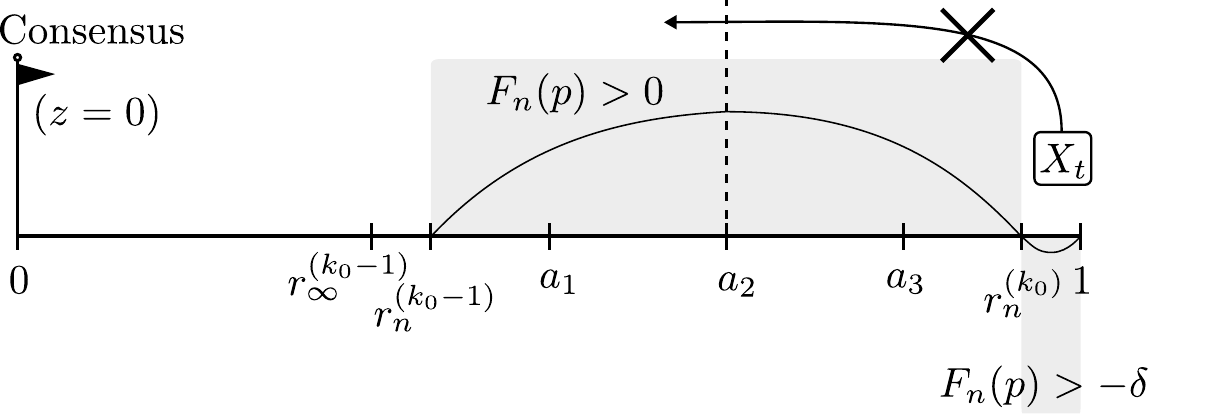}
    \caption{\textbf{Illustration of the arguments for Case 2.}
    We consider a configuration in which the correct opinion is~$0$.
    Constants $a_1, a_2$ and $a_3$ are chosen arbitrarily in the interval~$(\rinf^{(k_0-1)},1)$.
    By assumption, $F_n > 0$ on $[a_1,a_3]$. Moreover, once~$a_2$ and $a_3$ are fixed, we give a lower-bound on~$F_n$ on the interval~$[r_n^{(k_0)},1]$, by letting $r_n^{(k_0)}$ be sufficiently close to~$1$, in order to ensure that $X_t$ cannot jump from above $a_3 \, n$ to below $a_2 \, n$.
    Eventually, we are able to apply \Cref{cor:main_lemma_reversed}.
    }
    \label{fig:proof_sketch_positive}
\end{figure}

\paragraph*{Case 2.}
Let~$a_1,a_2,a_3 \in (\rinf^{(k_0-1)},1)$, with $\rinf^{(k_0-1)}<a_1<a_2<a_3<1$.
First, we show that by taking~$n$ large enough, we can have~$F_n$ be arbitrarily close to $0$ on the interval~$[r_n^{(k_0)},1]$.
\begin{claim} \label{claim:bound_on_F}
    For every~$\delta > 0$, 
    for $n$ large enough: $r_n^{(k_0)} > (1+a_3)/2$,
    and for every~$p \in [r_n^{(k_0)},1]$, $F_n(p) > - \delta$.
\end{claim}
\begin{proof}
    Let~$\delta >0$ and $n \in \support_2$.
    Since $F_n$ has bounded coefficients and degree~$d$, and since $F_n(r_n^{(k_0)}) = F_n(1) = 0$, by \Cref{claim:polynomial} in \Cref{sec:missing_proofs} we obtain the existence of~$C_0$ s.t.
    \begin{equation} \label{eq:bound_on_F}
        \text{for every } p \in [r_n^{(k_0)},1], \quad |F_n(p)| < C_0 \cdot (1-r_n^{(k_0)}).
    \end{equation}
    By \Cref{eq:root_limit} and by definition of~$k_0$, $r_n^{(k_0)}$ tends to~$1$ as $n$ goes to $+\infty$.
    If $n$ is large enough,
    \begin{equation} \label{eq:def_rn_k0}
        r_n^{(k_0)} > \max \left\{ 1-\frac{\delta}{C_0} , \frac{1+a_3}{2} \right\}.
    \end{equation}
    By \Cref{eq:bound_on_F}, this implies that $|F_n(p)|< \delta$ on $[r_n^{(k_0)},1]$, which, together with \Cref{eq:def_rn_k0}, concludes the proof of \Cref{claim:bound_on_F}.
\end{proof}

Now, we use the previous result to establish a lower bound on~$p + F_n(p)$ when~$p \geq a_1$.
\begin{claim} \label{claim:trapped_on_the_side}
    For~$n$ large enough,
    \begin{itemize}
        \item for every~$p \in [a_1,a_3]$, $p+F_n(p) > p$.
        \item for every~$p \in [a_3,1]$, $p+F_n(p) > a_3$.
    \end{itemize}
\end{claim}
\begin{proof}
    For $p\in [a_1,r_n^{(k_0)})$, $F_n(p) > 0$ by assumption.
    Therefore:
    \begin{itemize}
        \item for every~$p \in [a_1,a_3]$, $p+F_n(p) > p$.
        \item for every~$p \in [a_3,r_n^{(k_0)})$, $p+F_n(p) > p \geq a_3$.
    \end{itemize}
    All is left to prove is that $p+F_n(p) > a_3$ on $[r_n^{(k_0)},1]$.
    Let~$\delta = (1-a_3)/2$, and let $n$ be large enough for \Cref{claim:bound_on_F} to hold w.r.t. $\delta$. For every~$p \in [r_n^{(k_0)},1]$, we have
    \begin{align*}
        p+F_n(p) &> r_n^{(k_0)} - \delta & \text{(by \Cref{claim:bound_on_F} and definition of $p$)} \\
        &= a_3 + \pa{\frac{1-a_3}{2} - \delta} + \pa{r_n^{(k_0)} - \frac{1+a_3}{2}} & \\
        &> a_3, & \text{(by \Cref{claim:bound_on_F} and definition of~$\delta$)}
    \end{align*}
    which concludes the proof of \Cref{claim:trapped_on_the_side}.
\end{proof}
Finally, similarly to the first case, we use \Cref{cor:main_lemma_reversed} to conclude. Again, we start by checking that all assumptions hold. Let~$n$ be large enough for \Cref{claim:trapped_on_the_side} to hold.
    \begin{itemize}
        \item For every~$x_t \in \{a_1 \, n, \ldots, a_3 \, n\}$, we have~$x_t/n + F_n(x_t/n) > x_t/n$ by \Cref{claim:trapped_on_the_side}. Therefore, \Cref{eq:expectation_bound_z0n} rewrites
        \begin{equation*}
            \E \pa{ X_{t+1} \mid X_t = x_t } \geq x_t + n \, F_n\pa{\frac{x_t}{n}} - 1 \geq x_t - 1,
        \end{equation*}
        so assumption (i) in the statement of \Cref{cor:main_lemma_reversed} holds.
        
        \item For every~$x_t \in \{a_3 \, n, \ldots, n\}$, we have~$x_t/n + F_n(x_t/n) > a_3$ by \Cref{claim:trapped_on_the_side}. Therefore, \Cref{eq:expectation_bound_z0n} rewrites
        \begin{equation*}
            \E \pa{ X_{t+1} \mid X_t = x_t } \geq x_t + n \, F_n\pa{\frac{x_t}{n}} - 1 \geq a_3 \,n -1.
        \end{equation*}
        Therefore, by \nameref{thm:additive_chernoff_bound},
        \begin{align*}
            \Pr \pa{ X_{t+1} < a_2 \, n \mid X_t = x_t } &\geq \Pr \pa{ X_{t+1} < \E \pa{X_{t+1}} - \frac{a_3-a_2}{2} n \mid X_t = x_t } \\
            &\leq \exp\pa{-2\pa{\frac{a_3-a_2}{2}}^2 n},
        \end{align*}
        and so assumption (ii) holds.
        
        \item Finally, assumption (iii) follows from \nameref{thm:additive_chernoff_bound}, as in the first case.
    \end{itemize}
    Assuming that the source has opinion~$z=0$, we apply \Cref{cor:main_lemma_reversed}, which implies the existence of an initial configuration~$C_n$ s.t. the convergence time is bounded:   
    \begin{equation*}
         \tau_n(g,C_n) =  \inf \{t \in \bbN, X_t = 0\} \geq \inf \{t \in \bbN, X_t \leq a_1 n\} \geq n^{1-\varepsilon}.
    \end{equation*}
\end{proof}

\section{Discussion and Future Works}

In this paper, we explore the minimal requirements for simultaneously reaching consensus and propagating information in a distributed system.
We consider memory-less and anonymous agents, which update their opinion synchronously after observing the opinions of a few other agents sampled uniformly at random, and whose goal is to converge on the correct opinion held by a single ``source'' individual.
In addition, we adopt the self-stabilizing framework, which in a memory-less setting, means that convergence must happen for any possible initialization of the opinions of the agents (including the source).
Under this model, we show that to obtain a convergence time better than $n^{1-\epsilon}$, the number~$\ell$ of samples obtained by each agent in every round must necessarily tend towards infinity as $n$ increases.
Our result extends the range of values of~$\ell$ for which the performance of the ``minority'' dynamics (\Cref{prot:minority}) is characterized.
Our technique, which consists in translating the sample size into the degree of a well-chosen polynomial, and then inspecting its roots, is simple yet quite novel (to the best of our knowledge), and may be used to show similar results in other settings.

The ultimate goal of our work is to fully characterize the complexity of the bit-dissemination problem in the parallel setting and in the absence of memory, as a function of the sample size.
Regarding values of~$\ell$ allowing poly-logarithmic convergence time, there is still a large gap between our lower bound $\ell = \Omega(1)$ and the upper bound $\ell = O(\sqrt{n\log n})$ mentioned in~\cite{becchetti_minority_2024}.
Closing or narrowing this gap, even specifically for the minority dynamics, would be of appreciable interest in our opinion.

Another natural continuation would be to generalize our result to protocols using a constant amount of memory. If feasible, the resulting lower bound would still be compatible with the algorithm of~\cite{korman_early_2022}, which requires $\Omega(\log\log n)$ bits of memory.

\section*{Acknowledgements}

This work has been supported by the AID INRIA-DGA project n°2023000872 ``BioSwarm''.
The authors would like to thank Andrea Clementi for his feedback, and Amos Korman for helpful comments and coming up with the minority dynamics. Thanks also to Luca Trevisan for preliminary discussions, and for everything else.



\bibliography{main}

\clearpage
\appendix

\section{Well-known Dynamics} \label{sec:dynamics}

In this section, we define two important dynamics. \Cref{prot:voter,prot:minority} below only describe the behaviour of non-source agents, since the source never changes its opinion.
\vspace{1em}

\begin{minipage}{0.47\textwidth}
\begin{algorithm} [H]
    \DontPrintSemicolon
    \caption{Voter dynamics}
    \label{prot:voter}
    \KwIn{An opinion sample $S$ of size~$\ell$.}
    $X_{t+1}^{(i)} \leftarrow $ a random opinion in~$S$ \;
    \vspace{0.8em}
    \textit{(Since~$S$ is already sampled uniformly at random, this protocol yields the same dynamics regardless of the value of~$\ell$, and is usually defined with~$\ell = 1$).}
\end{algorithm}
\end{minipage}
\hfill
\begin{minipage}{0.47\textwidth}
\begin{algorithm} [H]
    \DontPrintSemicolon
    \caption{Minority dynamics~\cite{becchetti_minority_2024}}
    \label{prot:minority}
    \KwIn{An opinion sample $S$ of size~$\ell$.}

    \uIf{all opinions in~$S$ are equal to~$x$}{
        $X_{t+1}^{(i)} \leftarrow x$ \;
    }
    \Else{
        $X_{t+1}^{(i)} \leftarrow $ minority opinion in~$S$ \;
        \textit{(breaking ties randomly).}
    }
\end{algorithm}
\end{minipage}
\vspace{1em}

\noindent In terms of our definition, the Voter dynamics writes
\begin{equation} \label{eq:def_voter}
    g_n^{[0]}(k) = g_n^{[1]}(k) =: g^{\textrm{voter}}(k)=\frac{k}{\ell}, \quad \text{ for every }k\in \{ 0,\ldots,\ell\}.
\end{equation}

Similarly, given that ties are broken u.a.r., the Minority dynamics is given by
\begin{equation} \label{eq:def_minority}
    g_n^{[0]}(k) = g_n^{[1]}(k) =: g^{\textrm{minority}}(k)=
    \begin{cases} 
        1 \quad \text{if } k=\ell \text{ or } 0<k< \frac{\ell}{2},\\
        \frac{1}{2} \quad \text{if } k= \frac{\ell}{2},\\
        0 \quad \text{if } k=0 \text{ or } \frac{\ell}{2}<k<\ell.\\
    \end{cases}
\end{equation}

\section{Probabilistic Tools} \label{sec:probabilistic_tools}

\begin{theorem} [Hoeffding's bound] \label{thm:additive_chernoff_bound}
    Let~$X_1,\ldots,X_n$ be i.i.d. random variables taking values in $\{0,1\}$, let~$X = \sum_{i=1}^n X_i$ and~$\mu = \E(X) = n \Pr(X_1 = 1)$. Then it holds for all~$\delta > 0$ that
    \begin{equation*}
        \Pr \pa{ X \leq \mu - \delta } , \Pr \pa{ X \geq \mu + \delta } \leq \exp \pa{ -\frac{ 2 \delta^2}{n} }.
    \end{equation*}
\end{theorem}

The following version of Azuma's inequality, which accounts for the possibility that the martingale makes a large jump with a small probability, appears in~\cite[Section 8, p.34]{chung_concentration_2006}.
\begin{theorem} [Azuma-Hoeffding inequality] \label{thm:generalized_azuma_hoeffding}
    Let~$(X_t)_{t \in \bbN}$ be a martingale, and let $T \in \bbN$.
    If there is $p > 0$ and $c_1,\ldots,c_T$ such that
    \begin{equation*}
        \Pr \pa{ \exists t \leq T, X_t - X_{t-1} > c_t } \leq p,
    \end{equation*}
    then for every~$\delta>0$,
    \begin{equation*}
    	\Pr \pa{|X_T - X_0|> \delta} \leq 2 \exp \pa{-\frac{\delta^2}{2 \sum_{t=1}^T c_t^2}} + p.
    \end{equation*}
\end{theorem}

\section{Upper Bound for the Voter Dynamics} \label{sec:voter_upper}

The following proof involves only classical arguments (see, for instance, \cite{yildiz_binary_2013a} and \cite[Section 2.4]{hassin_distributed_2001}).
We nonetheless present it here for the sake of completeness.

\begin{proof} [Proof of \Cref{thm:secondary}]
The reader who is not already familiar with the idea is strongly encouraged to refer to \Cref{fig:dual_process} for an illustration.

Without loss of generality, we consider the Voter dynamics with~$\ell = 1$. In this case, the sample~$S_t^{(i)}$ of a non-source agent~$i \neq 1$ in round~$t$ is simply an element of the set~$I = \{1,\ldots,n\}$ of all agents, drawn uniformly at random.
In the case of the source, for the sake of the argument, we let~$S_t^{(1)} = 1$ for every~$t \in \bbN$, i.e., we consider that the source agent applies the Voter rule but always samples itself.

Now, let us fix an horizon~$T$. We will proceed by examining $n$ random walks~$\{W_{T-t}^{(i)}\}_{t \leq T}$, defined on the same randomness, but for which the time flows backward.
Specifically, for every~$i \in I$, let~$W_T^{(i)} = i$, i.e., every random walk ``starts'' at a different position.
Moreover, for $t<T$, let
\begin{equation*}
	W_t^{(i)} := S_t^{\pa{W_{t+1}^{(i)}}}.
\end{equation*}
In other words, if a random walk is in position~$j$ in round~$t+1$, and if Agent~$j$ samples~$i$ in round~$t$, then the random walk ``moves'' to~$i$ in round~$t$.
Note that by definition, if a random walk moves to position~$1$ in round~$t$, it will remain in position~$1$ for all remaining rounds:
\begin{equation} \label{eq:sink}
	W_t^{(i)} = 1 \implies \forall s \in \{1,\ldots,t\}, W_s^{(i)} = 1. 
\end{equation}
Moreover, if the random walk indexed by~$i$ ends up in position~$1$ (in round~$1$), it implies that the agent of index~$W_t^{(i)}$ holds the correct opinion in round~$t$:
\begin{equation} \label{eq:correctness}
	W_1^{(i)} = 1 \implies X_t^{\pa{W_t^{(i)}}} = z.
\end{equation}
To show that, we can simply proceed by induction on~$t$, using the definition of the random walks.
\Cref{eq:sink} and \Cref{eq:correctness} together implies that if the random walks indexed by~$i$ ever moves to position~$1$, then Agent~$i$ has the correct opinion in round~$T$:
\begin{equation} \label{eq:dual}
	\exists t \leq T, W_t^{(i)} = 1 \implies X_T^{(i)} = z.
\end{equation}
Therefore, for every~$i \neq 1$,
\begin{align*}
	\Pr \pa{ X_T^{(i)} \neq z } &\leq \Pr \pa{ \forall t \leq T, W_t^{(i)} \neq 1 } & \text{(by \Cref{eq:dual})} \\
	&\leq \prod_{t=0}^{T-1} \Pr \pa{ S_t^{\pa{W_{t+1}^{(i)}}} \neq 1 \mid W_{t+1}^{(i)} \neq 1 } & \text{(by the chain rule)} \\
	&= \pa{1-\frac{1}{n}}^T. &\text{(uniform and independent samples)}
\end{align*}
Note that this bound holds trivially for~$i=1$.
Therefore, by the union bound,
\begin{equation*}
	\Pr \pa{ \forall i \in I, X_T^{(i)} = z } = 1 - \Pr \pa{\exists i \in I, X_T^{(i)} \neq z} \geq 1 - \sum_{i \in I} \Pr \pa{ X_T^{(i)} \neq z} \geq 1 - n \pa{1-\frac{1}{n}}^T.
\end{equation*}
Taking~$T = 2 n \log n$, for $n$ large enough,
\begin{equation*}
    \pa{1-\frac{1}{n}}^{2 n \log n} = \exp \pa{2n\log n \log \pa{1-1/n}} \leq \exp \pa{-2 \log n} = \frac{1}{n^2},
\end{equation*}
which concludes the proof of \Cref{thm:secondary}.
\end{proof}

\begin{figure} [htbp]
    \centering
    \includegraphics[width=0.55\linewidth]{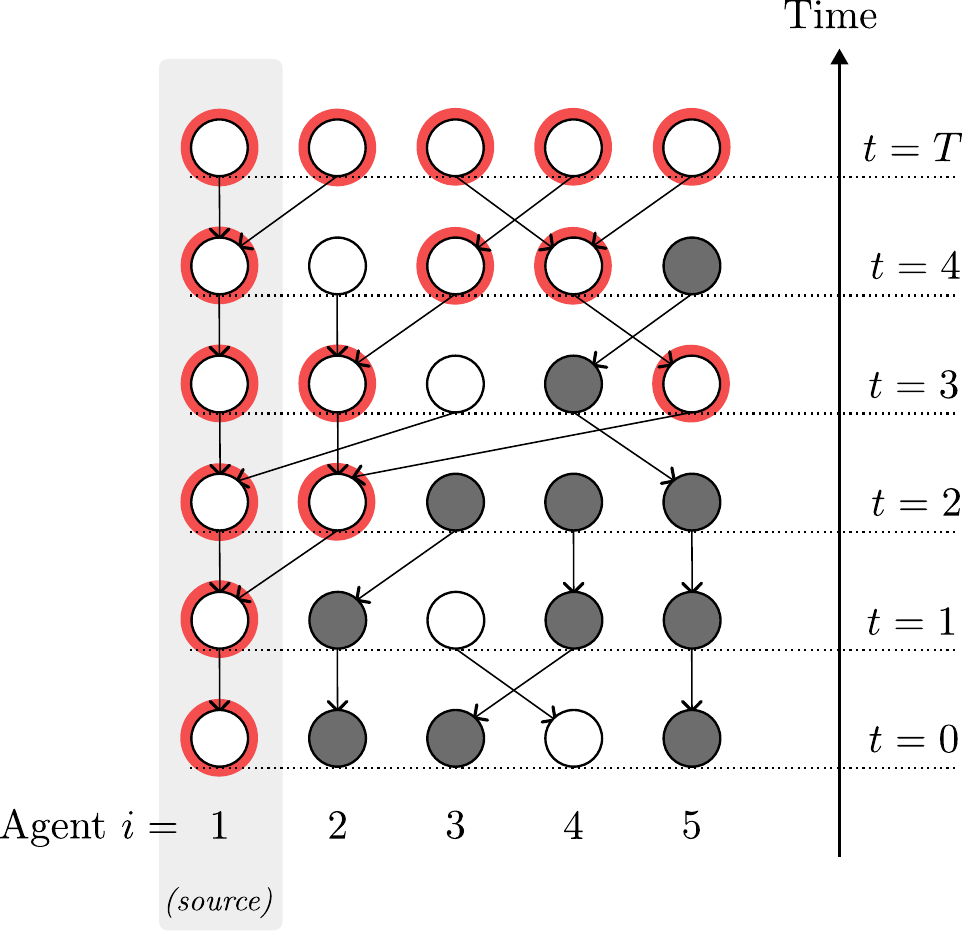}
    \caption{\textbf{Depiction of the dual process behind the proof of \Cref{thm:secondary}.}
    The color of the circle in row~$t$, column~$i$, corresponds to the opinion of Agent~$i$ in round~$t$: it is black if~$X_t^{(i)} = 1$, and white otherwise.
    An arrow is drawn from $(i,t+1)$ to $(j,t)$ if $S_t^{(i)} = j$, i.e., if Agent~$i$ observes Agent~$j$ in round~$t$ (and thus adopts their opinion in round~$t+1$).
    Red circles depict the locations of~$n$ coalescing random walks going backward in time, and initially present at every location. Random walks at a location~$i > 1$ make a move using the same randomness as the samples, while the source acts like a sink.
    If all random walks have coalesced in less than~$T$ rounds, it implies that the opinion of each agent in round~$T$ comes from the source, and thus that the dynamics has reached consensus on the correct opinion.
    }
    \label{fig:dual_process}
\end{figure}

\section{Missing Proofs} \label{sec:missing_proofs}

\begin{claim} \label{claim:polynomial}
    For every~$M,d$, there exists $C_0 = C_0(M,d) > 0$ s.t. for every polynomial~$P$ of degree $d$ and coefficients bounded by~$M$, every~$a,b \in [0,1]$ with $P(a) = P(b) = 0$, and every~$x \in [a,b]$, $P(x) < C_0 \cdot (b-a)$.
\end{claim}
\begin{proof}
    Since $P$ has degree $d$ and coefficients bounded by~$M$, there exists~$C = C(M,d)$ s.t. $|P'(x)| < C$ on $[0,1]$. Therefore, for every~$x \in [a,(a+b)/2]$, we have
    \begin{equation*}
        |P(x)| = |P(x) - P(a)| < C \cdot (x-a) < C \cdot \frac{b-a}{2}.
    \end{equation*}
    Similarly, for every~$x \in [(a+b)/2,b]$, we have
    \begin{equation*}
        |P(x)| = |P(b) - P(x)| < C \cdot (b-x) < C \cdot \frac{b-a}{2}.
    \end{equation*}
    Taking~$C_0 = C/2$ concludes the proof of \Cref{claim:polynomial}.
\end{proof}

\begin{proof}[Proof of \Cref{thm:voter_lower}]
Let~$\varepsilon>0$.
We want to apply \Cref{lem:main} with~$a_1=1/4$, $a_2=1/2$, $a_3=3/4$. If the three hypotheses hold, we have for the initial configuration $C_n := (z = 1, X_0 = \frac{a_2\,n+a_3\,n}{2})$,  
\[
    \tau_n \pa{\prot, C_n} \geq \inf\{t\in \bbN, X_t>a_3 n\}\geq n^{1-\varepsilon}.
\]
Now, let us show that the hypotheses hold:
\begin{itemize}
    \item \textbf{Proving (i).}
    Since $F_n = 0$ for $n$ large enough, we have by \Cref{lem:expectation_bounds} that $\E(X_{t+1} \mid X_t = x_t) \leq x_t+1$.
    \item \textbf{Proving (iii).} Conditioning on $X_t$, $X_{t+1}$ is a sum of~$n$ independent Bernoulli random variables. Therefore, we can use \nameref{thm:additive_chernoff_bound} to obtain
    \[
        \Pr(|X_{t+1}-\E\pa{X_{t+1} \mid X_t}| > n^{1/2 + \varepsilon/4}) < 2\exp \pa{ -2n^{\varepsilon/2}},
    \]
    which establishes (iii). 
    \item \textbf{Proving (ii).} If $x_t<a_1\, n$ and $n$ is large enough, (ii) follows again from \nameref{thm:additive_chernoff_bound}:
    \begin{align*}
    	\Pr(X_{t+1} > a_2 n \mid X_{t}=x_t) &\leq \Pr(X_{t+1} > n^{1/2 + 1/4} +\E\pa{X_{t+1} \mid X_t=x_t} \mid X_{t}=x_t) \\
    	&< 2\exp \pa{ -2n^{1/2}},
    \end{align*} 
\end{itemize}
which establishes (ii) and concludes the proof of \Cref{thm:voter_lower}.
\end{proof}

\end{document}